\documentclass[preprint,3p,twocolumn, a4paper,top=3cm,bottom=3cm,left=3cm,right=2.7cm, marginparwidth=2.5cm,asymmetric]{elsarticle}

\usepackage[utf8]{inputenc}
\usepackage[T1]{fontenc}
\usepackage{scrextend}

\usepackage{graphicx}
\graphicspath{{Images/}}
\usepackage{geometry}
\usepackage{mathtools}
\usepackage{mathrsfs}
\usepackage{amssymb}
\usepackage{amsmath}
\usepackage{amsthm}
\usepackage[hidelinks]{hyperref}
\usepackage{comment}
\usepackage{anyfontsize}
\usepackage{amssymb}
\usepackage{caption}
\usepackage{cellspace}
\usepackage{booktabs}
\usepackage{arydshln}
\usepackage{amssymb}
\usepackage{enumerate}
\usepackage{physics}
\usepackage{stmaryrd}
\usepackage{emptypage}
\usepackage{setspace}
\usepackage{datetime}
\usepackage{fontenc}
\usepackage{array} % Per personalizzare la formattazione delle colonne
\usepackage[table]{xcolor} % Per aggiungere colori alle tabelle
\usepackage{natbib}

\newtheorem{lemma}{Lemma}[section]
\newtheorem{theorem}[lemma]{Theorem}

\newtheorem{remark}[lemma]{Remark}

\newtheorem{definition}[lemma]{Definition}

\newcommand{\cC}{\mathcal{C}}
\newcommand{\cH}{\mathcal{H}}
\newcommand{\cR}{\mathcal{R}}
\newcommand{\cB}{\mathcal{B}}
\newcommand{\cS}{\mathcal{S}}

\newcommand{\cQ}{\mathcal{Q}}
\newcommand{\ns}{\mathcal{N}}
\renewcommand{\r}{\mathbb{R}}
\newcommand{\n}{\mathbb{N}}
\newcommand{\z}{\mathbb{Z}}

\begin{document}




\begin{frontmatter}
\title{Universality conditions of unified classical and quantum reservoir computing}

\author[lab]{Francesco Monzani}
%\ead{francesco.monzani@unimi.it}
\author[lab]{Enrico Prati\corref{cor1}}
\ead{enrico.prati@unimi.it}

\affiliation[lab]{organization={Dipartimento di Fisica A. Pontremoli, Universita degli Studi di Milano},
            addressline={Via Celoria 16}, 
            city={Milano},
            postcode={20139}, 
            country={Italy}}

%\affiliation{Consiglio Nazionale delle Ricerche, Piazzale Aldo Moro 7, I-00185 Roma, Italy}
\cortext[cor1]{Corresponding author.}
\begin{abstract}
Reservoir computing is a versatile paradigm in computational neuroscience and machine learning, that exploits a recurrent neural network to efficiently process time-dependent information.
The power of many neural network architectures resides in their universality approximation property. As widely known, classes of reservoir computers serve as universal approximators of functionals with fading memory. The construction of such universal classes often appears context-specific, but, in fact, they follow the same principles. Here we present a unified theoretical framework and we propose a ready-made setting to secure universality, based on the minimal sufficient conditions for a class of reservoir computers to be universal, namely the fading memory and the polynomial algebra structure of the set of their associated functionals.
We test the result in the arising context of quantum reservoir computing. Guided by such a unified theorem we suggest why spatial multiplexing serves as a computational resource when dealing with quantum registers, as empirically observed in specific implementations on quantum hardware. The analysis sheds light on a unified view of classical and quantum reservoir computing.
\end{abstract}

\end{frontmatter}

\section{Introduction}
Reservoir computing is a general computational framework that exploits the nonlinear dynamics of a recurrent neural network and simple readout functions for the online processing of time-dependent inputs \cite{naka.2019, spec.issue}. The spread of reservoir computing has stemmed from two primary paradigms, namely echo state networks (ESN), for time series prediction, \cite{jaeger.first.long,  esn.jaeger.haas} and liquid state machines (LSM) respectively, the latter inspired by spiking neuronal activity for modeling online time-dependent inputs \cite{maass.nat, maass.sontag.neural}.\\
A computational architecture based on neural networks is said to be universal if for any given target function there exists an instance of the architecture realizing a mapping between the input and the output that approximates the given function with arbitrary precision. The issue of finding neural networks that serve as universal approximators has attracted considerable mathematical research during the years \cite{kolmo, boyd.chua, sand, matt, cuck, cucker_zhou_2007}, paving the way for the implementation of versatile machine learning paradigms. Universal approximation results have been proven for different neural network architectures, starting from the classical results for feed-forward \cite{ hornik.univ.fnn, cyb, hornik.3, hornik.4, spre1, spre2, stinchcombe1999neural} and for recurrent neural networks, respectively \cite{HAMMER2000107, rnn.univ}. More recently, universality has been proven also for invertible \cite{jin2024approximation, ishikawa2024universal},   for mean-field \cite{pham2023meanfield}, and for deep convolutional neural networks \cite{zhou2020universality}. \\
%A variety of reservoir computer architecture also quantum
Restricting to reservoir computing, both ESN \cite{grig2017,grig2018,gonon} and LSM \cite{mass.univ} have yielded classes of reservoir computers that serve as universal approximators of functionals with fading memory. Despite the similarity in the strategy exploited in the proofs of these results, ultimately based on an application of the Stone-Weierstrass theorem in the spirit of the seminal work of Boyd and Chua in 1985 \cite{boyd.chua}, they appeared as context-specific.\\
Emerging from these frameworks, as anticipated by Maass \cite{maass.motivation}, reservoir computing has proved its efficiency not restricted to a digital approach. In the last twenty years, numerous physical implementations of reservoir systems have been proposed as valid tools for reservoir computing \cite{survey, tanaka2019recent, markovic2020physics, Nakajima_2020, physi.skir}, exploiting for example, optical devices \cite{ph3,ph2, ph4,ph1}, spiking neurons and cortical networks \cite{DOCKENDORF200990, yada2021physical, spike, PENG2024274} and nanoscale oscillators \cite{torrejon2017neuromorphic, papp2021nanoscale}. More recently, the advent of quantum computation \cite{prati2017quantum,de2023silicon} has also shed light on the development of reservoir computers harnessing the high computational capacity of quantum reservoir systems. After the first proposal \cite{giappo.1}, many quantum systems have been proposed as reliable effective reservoirs \cite{ kutvonen2020optimizing, govia,  martinez2021dynamical, mujal2021opportunities, suzuki.quantum,  qrc, mujal.quantu, molteni2023optimization, ghosh2021reconstructing,garcia2023scalable, domingo2020taking} and even as examples of universal classes of reservoir computers \cite{chen1, chen2, nokkala2021gaussian}. In particular, quantum networks have been proven effective in processing physical quantum information and predicting quantum dynamics \cite{ghosh2019quantum, ghosh2021reconstructing, ghosh2021realising, lazzarin2022multi}.
As new embodiments of reservoir computing are emerging after the advent of quantum computing, a unified framework is needed in order to assign universality. This work aims to define such a generalized framework to determine which minimum set of properties is mandatory to be fulfilled by a family of reservoir computers to exhibit universality. A unified theoretical approach, that includes systems fed with inputs living in general normed space, may act as a guideline promoting the development of new implementations of reservoir systems. \\
Due to the large variety of reliable classes of reservoir computers, we propose a unified theoretical framework - that we call \emph{echo} reservoir computers - together with a ready-made setting that guarantees universality. We then show that echo state networks and liquid state machines fall within our unified framework and we discuss an example that entails a quantum reservoir \cite{naka.2019}. 
 \\
It results that the sufficient conditions for a class of reservoir computers to be universal are respectively that the set of associated functionals has continuity with respect to a fading metric and is a polynomial algebra. The fading metric accounts for the fading memory property by combining the sup-norm with a monotonic increasing function in $\r^-$, referred to as the fading function. In this perspective, we test such a unified framework, showing that it encompasses the two most notable embodiments of reservoir computers (ESN and LSM). Moreover, we prove that our framework is general, including an example concerning a gate-based quantum reservoir. In this context, spatial multiplexing, namely the parallel processing of information with disjoint registers and inter-depending readout functions, accounts for universality, as empirically recognized in \cite{naka.2019} and theoretically confirmed in \cite{chen1,chen2,nokkala2021gaussian}. \\
The rest of the paper is organized as follows: in Section 2 we present the theoretical framework, by defining a general notion of reservoir computer, stating the main general theorem on universality and discussing the notion of fading memory. In Section 3 we retrieve classical results on the universality of families of echo state networks and liquid state machines. In Section 4, we discuss the notion of quantum echo state network. Furthermore, we discuss an example of a universal class of quantum reservoir systems.
\section{Theoretical framework}\label{termino}
A reservoir computer is a computational architecture for the online processing of temporal data, either discrete or continuous in time. In both cases, we refer to them as orbits, using a terminology borrowed from dynamical systems theory. A reservoir computer produces a mapping between input and output orbits by the composition of a recurrent network and a readout function. We will refer to any mapping between orbits as a filter. In reservoir computing, at each temporal step, the current value of the input orbit is encoded in the dynamics of the recurrent network or the physical substrate that acts as a reservoir. Subsequently, the current state of the reservoir is processed by a readout function, usually depending on some nodes in the reservoir often referred to as the true nodes, that produces the value of the output orbit at the current time step.   
We remark that when a reservoir computer is used for learning a given task, thus for approximating a given filter, the sole trained component is the readout function, minimizing the distance between the mapping and the value of a given target filter.
\begin{figure}
\begin{center}
\includegraphics[scale = .38]{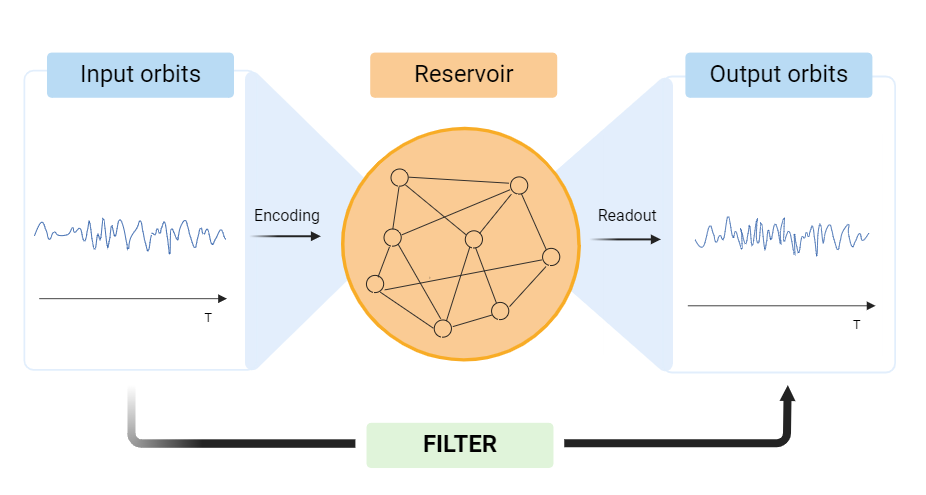}
\end{center}
\caption{\small{A schematic representation of a reservoir computer for one-dimensional orbits ($n=1$). An input orbit is encoded in the reservoir that acts as a hidden layer. The response of the reservoir is used to produce an output through a trainable readout mapping. Here T can be either $\r$ or $\z$.}}
\end{figure} 
\subsection{Definitions of the input data and the reservoir}
In this section, we introduce the definitions and some properties of filters and functionals. In particular, we carry out a notation suitable for both discrete and continuous time evolution of the reservoir systems. All throughout the paper, $T$ will indicate both $\r$ or $\z$. 
\begin{definition}
    Let $T = \r$ or $T = \z$. Given $\left(\ns,\norm{\cdot}\right)$ a normed space, we call \textbf{orbit} any function of the form 
    \begin{align*}
        u\colon T \rightarrow \ns.
    \end{align*}
    We denote with $\ns^T$ the space of orbits in $\ns$ labeled in $T$.
\end{definition}
We will consider only sets of uniformly bounded orbits. Namely, we define the following space.
\begin{definition}\label{bounded}
   Given a real $K>0$ and a normed space $\left(\ns,\norm{\cdot}\right)$, we denote with
    \begin{align*}
        S^{T, \ns}_{K} \coloneqq \left\{u\colon T \rightarrow\ns: \norm{u_t}\le K \quad \forall
        t\in T \right\}
    \end{align*}
    the set of uniformly bounded orbits on $\ns$. 
\end{definition}
Usually, classical reservoir computers take inputs belonging to subsets of $S^{T,\r^n}_{K}$. Examples of typical subsets that we will consider later on are given by continuous and Lipschitz functions, spike trains, or time series. These subsets are compact metric spaces if endowed with a proper distance. Moreover, the case of reservoir systems that take pure quantum states is included in our definition of orbits, taking $\ns$ as a complex Hilbert space. 
\begin{definition}
        Given $\tau\in T$, we denote with $u^\tau$ the time-delayed orbit defined by 
\begin{align*}
    (u^\tau)_t \coloneqq u_{t-\tau}\,.
\end{align*}
\end{definition}
We call filter any map between sets of orbits and functional any map that sends an orbit in a scalar value. Precisely, we have the following definitions. 
\begin{definition}
    We call \textbf{filter} any map of the form
    $B \colon\ns_1 ^T\rightarrow \ns_2^T$, for some normed spaces $\ns_1,\ns_2$.
   \begin{itemize}
        \item[i.] A filter is \textbf{causal} if for any given $\tau\in T$, $u_t = v_t \,\, \forall t\le\tau$ implies
    \begin{align*}
        (Bu)_\tau =  (Bv)_\tau \,;
    \end{align*}
    \item [ii.] a filter is \textbf{time invariant} if for any $\tau\in T$, 
    \begin{align*}
       \left(Bu^\tau\right)_t = \left(Bu\right)_{t-\tau},\qquad \forall t \in T\,.
    \end{align*}
    \item  [iii.] a filter is \textbf{bounded} if there exists $K'>0$ such that  for any $u\in S^{T,\ns_1}_{K}$,
    \begin{align*}
        \norm{(Bu)_t}< K',\, \forall t\in T\,.
    \end{align*}
    \end{itemize} 
   We call BCTI filter a filter that is bounded, causal and time-invariant.
\end{definition}
\begin{definition}
    Given two normed spaces $\ns_1,\ns_2$, a \textbf{functional} from $\ns_1$ to $\ns_2$ is any map of the form $H\colon \ns_1^T \rightarrow \ns_2$.
\end{definition}
It is well known that every BCTI filter is associated with a unique functional. Precisely, we have the following lemma. We recall the proof for the sake of completeness.
\begin{lemma}\label{bijection}
    There is a bijection between the set of functionals and the set of BCTI filters. 
\end{lemma}
\begin{proof}
    Given a BCTI filter $B$, the associated functional is given by its value in $t=0$, namely $H_B(u) \coloneqq (Bu)_0$. Conversely, given a functional $H$, the associated filter is $B_H$ defined via the time-delayed orbit,  $(B_H u)_t \coloneqq H(u^{-t})$. Moreover, it is straightforward to verify that $ B \equiv B_{H_B} $ and $H \equiv H_{B_H}$. In fact, we have the equalities
  \begin{equation*}
      \left(B_{H_B} u\right)_t = H_B(u^{-t}) = \left(B u^{-t}\right)_0 = \left(B u\right)_t
  \end{equation*}
  since $B$ is time-invariant and 
  \begin{equation*}
  H_{B_H} (u) = \left(B_H u\right)_0 = H (u^0) = H(u)\,.
  \end{equation*}
\end{proof}
 As a consequence of the latter lemma, one is allowed to consider only functionals, which are simpler mathematical objects with respect to filters. In particular, it allows the application of Stone - Weierstrass theorem for filters and functionals with target in $\r$.
 \begin{remark}
We will refer to filters and functionals with $\ns=\r$ as target normed space as real filters and real functionals, respectively.
\end{remark}
\subsection{Universal echo reservoir computers}
Now we turn to the concept of echo reservoir computer, a general machine which unifies echo state networks and liquid state machines by a single class. The name reminds that such a system associates a unique output to any given input orbit. In the context of echo state networks, such property is often referred to as \textit{echo state property}. Later on, we show that, under proper conditions, we can write both liquid state machines and families of echo state networks as echo reservoir computers. 
\begin{definition}
    Let $\ns_1, \ns_2$ be two normed spaces. Let $E \colon \ns_1^T \rightarrow\ns_2^T$ be a given BCTI filter and $H\colon\ns_2\rightarrow\r$ a given real function. We call \textbf{echo reservoir computer} the system defined by
    \begin{align*}
    \begin{cases}
        x_t = \left(Eu\right)_t\\
        y_t = H(x_t),\qquad \forall t \in T\,.
    \end{cases}
    \end{align*}
   We denote this echo reservoir computer as $\cC = \cC(E,H)$.
\end{definition}
Usually, one refers to $x_t \in \ns_2 $ as the reservoir state at time $t$ and to $H$ as the readout function. Generally, one defines more general reservoir computers, without requiring causality and time-invariance. Here, we are assuming that echo reservoir computers are causal and time-invariant by definition since all the universality results pertain to systems with these properties. Moreover, we remark that, according to the above definition, any echo reservoir computer associates a unique output to any input. In the literature, such property is known as the \textit{echo state property}. The notion of a echo reservoir computer fulfills a standard compactness condition since we are assuming $E$ to be a bounded filter. Namely, if $u\in S^{T,\ns_1}_{K}$, then $x\in S^{T,\ns_2}_{K'} $ for some $K'>0$. \\ Any echo reservoir computer naturally defines a real BCTI filter $R_{(E,H)}$, that we call the reservoir filter associated with the echo reservoir computer $\cC(E,H)$, defined by 
\begin{align*}
   &y = R_{(E,H)}u\,,\\ &y_t = \left(R_{(E,H)}u\right)_t = h\left((Eu)_t\right)\quad \forall t\in G\,.
\end{align*}
Then, recalling that any BCTI filter determines a unique functional, we can associate to any echo reservoir computer a real functional, that for simplicity we still denote $R_{(E,H)}$.
\\ \\
In the following, we prove a general theorem about the density of real continuous functionals on compact metric spaces and we deduce from that the universality of families of echo reservoir computers. Later on, we will show that fading memory is the suitable continuity property to assure compactness and thus universality. Since we restrict on causal, time-invariant filters, we can consider orbits restricted to $\z_-$ and $\r_-$. 
\begin{definition}
    Denote with $T_-$ either $ \z_-$ or $\r_-$. We call \textit{restricted} orbit any function of the form
    $ u \colon T_- \rightarrow \r^n$.
    Accordingly to Def. \ref{bounded}, we denote with $S^{T,\ns}_{K,-}$ the set of all the bounded orbits and with $I_- \subset S^{T,\ns}_{K,-}$ any subset of inputs.
\end{definition}
\begin{remark}
The unique functional associated with a BCTI filter is uniquely determined by its action on restricted orbits. 
\end{remark}
A natural necessary condition for the universality of a class of functionals is to discriminate different inputs. 
\begin{definition}
    A family of functionals $\cH$ separates orbits in $I_-$ if for any pair of orbits $u,v\in I_-$ with $u\not=v$ there exists a functional $H\in \cH$ such that $H u \not = H v$.
\end{definition}
Before stating the main theorem, we need to specify the polynomial operations between functionals. More explicitly, it consists of defining sums and products of functionals. For simplicity, we restrict on real functionals, although these definitions can be generalized to functionals between arbitrary normed spaces.
\begin{definition}
 Given $u \in S^{T,\ns}_{K}$, 
 two real functionals $H_1, H_2$ and $\lambda \in \r$, we denote 
\begin{equation}\label{poly}
\begin{aligned}
    \left(H_1 + \lambda H_2\right) \coloneqq H_1 u + \lambda H_2 u \\
    \left(H_1 \cdot H_2\right)u \coloneqq  H_1 u \cdot H_2 u\,.
\end{aligned}
\end{equation}   
\end{definition}
Then, a polynomial function of some functionals $\left\{H_i\right\}_{i=1,\dots,l}$ is defined accordingly to Eq.\eqref{poly}. Precisely, we recall the notion of a polynomial algebra.
\begin{definition}
A family of real functionals $\cH$ is a \textbf{polynomial algebra} if for any $R_1,R_2 \in \cH$ and any $\lambda \in \r$ there exist $R^\lambda_{+}, R_{\cross} \in \cH$ such that
    \begin{equation*}
        R_1 + \lambda R_2 = R_+^\lambda\,,\quad R_1 \cdot R_2 = R_{\cross}\,,
    \end{equation*}
\end{definition}
We are in the position to state and prove our main theorem, which points out some sufficient conditions that imply the density of a family of real functionals.
\begin{lemma}\label{func.dens}
Let $ I_- \subset S^{T,\ns}_{K}$ be a subset of restricted orbits and assume that $\left(I_-, d\right)$ is a compact metric space. Let $\cH$ be a family of real functionals such that
\begin{itemize}
    \item separates orbits in $I_-$;
    \item contains the constant functionals;
    \item any functional $R\in \cH$ is continuous with respect to the metric $d$\,.
    \end{itemize}
    Then any functional $H\colon I_-\rightarrow \r$ continuous with respect to the metric $d$ is approximated with arbitrary precision by a polynomial function of some functionals in $\cH$. Precisely, for any $\epsilon>0$ there exist $\left\{R_i\right\}_{i=1,\dots,l}$ with $ R_i \in \cH$ and a polynomial $p$ such that $\left|p(R_1,\dots,R_l)u - H u\right| < \epsilon$ for any $u\in I_-\,$.
\end{lemma}
The proof is a straightforward application of the Stone-Weierstrass theorem, which states that, 
given $(W,d)$ a compact metric space and $S$ a subset of real continuous functions, $S \subset C(W,\r)$, if $S$ separates points and contains constant functions, then for any $\epsilon>0$ and any $f\in C(W,\r)$, there exist $\left\{s_i\right\}_{1,\dots,l},\, s_i \in S$ and a polynomial $p$ such that $|p(s_1,\dots,s_l)(x) - f(x)|< \epsilon$ for any $x \in W$.
\begin{proof}
By hypothesis, $\cH$ is a subset of the set of the real continuous functions on $\left(I, d_\omega\right)$, which is a compact metric space by hypothesis.
Moreover, $\cH$ separates points in $U$ and contains constant functions. Then, we apply the Stone-Weierstrass theorem and we get the thesis.
\end{proof}
The latter lemma does not provide alone the universality of a family of echo reservoir computers $\cR$, since it may be possible that the functional $p(R_1,\dots,R_l)$ is not associated with any of the elements in $\cR$. If the family of functional associated with $\cR$ is a polynomial algebra, then the conditions in Theorem \ref{func.dens} ensure that $\cR$ is universal.
\begin{theorem}[Universality]\label{univer}
Let $\cR$ be a family of echo reservoir computers and let $I$ be a set of input orbits. Assume that $(I_-,d)$ is a compact metric space. Assume that the set of functionals $ \cH_{\cR}$ associated with $\cR$ is a polynomial algebra that separates points, contains constant functionals, and such that each of its elements is continuous with respect to the metric $d$. \\ 
Then $\cR$ is a universal class of echo reservoir computers with respect to the inputs $I$. Namely, for any $\epsilon>0$, for any real, bounded functional $H\colon I\rightarrow \r$ continuous with respect to $d$, there exists a echo reservoir computer $\cC(E,h) \in \cR$ such that
    \begin{align*}
        \left|R_{(E,h)} u - H u\right| < \epsilon, \quad \forall u \in I,.
    \end{align*} 
    \end{theorem}
    \begin{proof}
        Let $H\colon I\rightarrow\r$ be a real bounded functional, continuous with respect to $d$. Then, from Lemma \ref{func.dens}, there exist some functionals $R_1,\dots,R_l \in \cH_{\cR}$ and a polynomial $p$ such that $|p(R_1,\dots,R_l) u - H|< \epsilon$ for any $u\in I_-$. Since $\cH_{\cR}$ is a polynomial algebra, then $p(R_1,\dots,R_l) \in \cH_{\cR}$; this means that there exists an echo reservoir computer $\cC(E,h)$ such that the functional associated to it is $p(R_1,\dots,R_l)$, namely $R_{(E,h)} = p(R_1,\dots,R_l)$.
    \end{proof}
\begin{remark}
    Recalling the one-to-one correspondence between BCTI filters and functionals, one can restate the universality property as the universal approximation of BCTI real filters.
\end{remark}
    \begin{remark}
        In the latter theorem, we ask the family of functionals associated with a given class of reservoir computers to be a polynomial algebra. In many examples, one exploits polynomial readouts to ensure this condition is verified, although is not necessary. See Refs. \cite{grig2017, grig2018} for examples of universal classes of reservoir computers with linear readout.
    \end{remark}
    \subsection{Continuity and fading memory}
 In the context of recurrent neural networks, the concept of fading memory plays a key role \cite{boyd.chua, sand.fad.mem}. Informally, a functional, or a filter, has fading memory if its dependence on older information is reduced over time. In the former subsection, we considered the sufficient conditions for the approximation of continuous functionals defined on inputs belonging to compact metric space. In this section, we show that fading memory provides suitable continuity to be assured by filters and functionals. We specify its formal definition for two different sets of input orbits, namely time series and spike train.
    \subsubsection*{\textbf{Time series}}
    We consider first time series, thus bounded orbits labeled in $T = \z$, namely $I_-^K = S^{\z,\ns}_{K}$. Fading memory ensures the compactness of the associated metric space.
    \begin{definition}\label{fad.ts}
        For any given $\omega\colon T_-\rightarrow(0,1]$ an increasing function with zero limit at infinity and $\omega(0)=1$, that we will call fading function, we define the \textbf{fading metric}
        \begin{align*}
        d_\omega(u,v)\coloneqq  \sup_{t\in T_-}\norm{u_t-v_t}\omega_{t}, \quad \forall u,v \in I_-^K \,.
    \end{align*}
    \end{definition}
    We have the following lemma.
    \begin{lemma}\label{metric.space}
    If the normed space $\ns$ is finite-dimensional the space $I_-^K$ equipped with the distance $d_\omega$ is a compact metric space. 
\end{lemma}
\begin{proof}
    By definition $d_\omega(u,u) = 0$, $d_\omega(u,v)>0$ for any $u\not=v$ and $d_\omega(u,v) = d_\omega(u,v)$. Triangular inequality follows since for any $u,v,z$
\begin{equation*}
    \begin{aligned}
        d_\omega(u,z) = \sup_{t\in T_-} \norm{u_t - z_t}\omega_t =\\= \sup_{t\in T_-} \norm{u_t - z_t + v_t - v_t}\omega_t \le \\
        \le \sup_{t\in T_-}\left(\norm{u_t - v_t} + \norm{v_t - z_t}\right) \omega_t \le \\ \le \sup_{t\in T_-} \norm{u_t - v_t}\omega_t + \sup_{t\in T_-} \norm{v_t - z_t}\omega_t = \\
        =   d_\omega(u,v) +  d_\omega(v,z)
    \end{aligned}
\end{equation*}
We prove that $\left(I_-^K,d_\omega\right)$ is compact by proving that one can extract a convergent subsequence from any sequence in it. \\
For convenience, for any $m\in\n$ and any $u=(u_{-\infty},\dots,u_{-1},u_0)\in I_-^K$ we denote with $u^m = (u_{-m},\dots,u_{-1},u_0)$ the time-restricted orbit and with $I_{-,m}^K$ the space of these time-restricted orbits. Then, any $I_{-,m}^K$ is compact, and thus sequentially compact, since it is bounded and closed in the finite-dimensional normed space $\ns^m$. \\ 
Given $\left\{u(n)\right\}_{n\in\n}$ a sequence in $ I_-^K$, we extract from it a subsequence $u(n_j)$ such that $d_\omega(u(n_j),\tilde{u})\rightarrow0$ for a certain $\tilde{u}\in I_{-}^K$. For any fixed $k\in\n$, $u^k(n)$ is a sequence in  $I_{-,k}^K$. Thus, since $I_{-,k}^K$ is sequentially compact, there exists a sequence $n_k$ such that 
\begin{equation*}
    \sup_{t\in\{-k,\dots,0\}}\norm{u_t^k(n_k)-\tilde{u}_t^k}\rightarrow 0, \quad n_k\rightarrow +\infty \,.
\end{equation*}
Reasoning as above for any $k\in \n$ and observing that one can choose any sequence $n_k$ such that $n_{k_2}$ is contained in $n_{k_1}$ for any $k_2>k_1$, we note that $\tilde{u}^{k_2}$ is the extension of $\tilde{u}^{k_1}$. Thus $\tilde{u}$ is a unique, well-defined element in $I_-^K$, such that there exists a sequence $(n_j)_{j\in\n}$ for which, for any $k_0\in\n$
\begin{align}\label{conv.compact}
\sup_{t\in\{-k_0,\dots,0\}}\norm{\tilde{u}_t - u_t(n_j)} \rightarrow 0, \quad j\rightarrow +\infty\,.
\end{align}
We come to prove that $d_\omega\left(u(n_j),\tilde{u}\right)\rightarrow 0 $ as $j\rightarrow +\infty$. Fix $\epsilon>0$. By definition of fading function $\omega$, there exists $T_0$ big enough such that for any $t\ge T_0$ 
\begin{equation}\label{tempi.lontani}
\begin{aligned}
    \sup_{t\le -T_0}\norm{u_t(n_j) - \tilde{u}_t}\omega_{-t}\le 2 K \omega_{T_0} \le \epsilon 
\end{aligned}
\end{equation}
for any $j\in \n$. Moreover, from \eqref{conv.compact}, we have that there exists $J\in\n$ such that for any $j\ge J$
\begin{equation}\label{tempi.vicini}
\begin{aligned}
\sup_{t\in\{-T_0,\dots,0\}}\norm{u_t(n_j) - \tilde{u}_t}\omega_{-t} \le \\
\sup_{t\in\{-T_0,\dots,0\}}\norm{u_t(n_j) - \tilde{u}_t} \le \epsilon \,.
\end{aligned}
\end{equation}
Thus, combining \eqref{tempi.lontani} and \eqref{tempi.vicini}, we have proven that $d_\omega\left(u(n_j),\tilde{u}\right)\rightarrow 0 $ as $j\rightarrow +\infty$.
\end{proof}
The fading memory condition is the continuity in the topology of this compact metric space.
    \begin{definition}
        A functional $H$ taking inputs in $I^K$ fulfills the fading memory condition if it is continuous in the topology of the metric space $\left(I^K_-, d_\omega\right)$.
    \end{definition}
     \begin{definition}
     A BCTI $B$ fulfills the fading memory condition if the unique associated functional $H_B$ has fading memory.   
    \end{definition}
    \begin{remark}
Fading memory is a continuity property, thus it is preserved by any continuous operation. In particular, fading memory is preserved by sums and products in a polynomial algebra.
\end{remark}
    
\subsubsection*{\textbf{Spike train inputs}}
A spike train is modeled by a sequence of spiking activities in time, that is a discrete subset $u\subset\r$. In particular, denoting the neural refractory period with $\Delta>0$, one can assume that each couple of elements in $u$ have a distance at least $\Delta$. Equivalently, one may describe a spike train $u$ by a function, that we still denote with $u\colon\r\rightarrow\{0,1\}$, defined by
\begin{align*}
    u_t = \begin{cases}
        0 \quad t \not\in u \\
        1 \quad t \in u\,.
    \end{cases}
\end{align*}
Precisely, we consider the following set of spike train inputs
    \begin{align*}
        I^\Delta= \{u\colon \r \rightarrow \left\{0,1 \right\} \colon u_t + u_s < 2, \\ \forall \,|t - s|\le \Delta\}\,.
    \end{align*}
    and the related set of restricted orbits $I^\Delta_-$.\\
    Given any fading function $\omega$, one can construct a distance $d^c_\omega$ for which fading memory functionals defined on $(I^\Delta, d^c_\omega)$ are continuous \cite{mass.univ}. To do this, we approximate any $u\in I^\Delta$ with a continuous function $f^u$. Explicitly, one can define such a function by
    \begin{align*}
        f^u \coloneqq \sum_{s\in u} T^s
    \end{align*}
    with $T^s$ a positive continuous function such that $\left(T^s\right)_s = 1$ and $\left(T^s\right)_t$ if $|t-s|\ge 1$. Then one can define a distance $d^c_\omega$ so that $\left(I^\Delta_-,d^c_\omega\right)$ is a compact metric space. 
\begin{definition}
    For any $u,v\in I^\Delta_-$ and any given fading function $\omega$, we define $$d^c_\omega(u,v)\coloneqq \int_{-\infty}^0 \left|\left(f^u\right)_t - \left(f^v\right)_t\right|\omega_{t}\, dt\quad.$$
\end{definition}
The following result is proved in Ref. \cite{mass.univ}.
\begin{lemma}\label{fm-spike}
    The space $I^\Delta_-$ equipped with the distance $d^c_\omega$ is a compact metric space.
\end{lemma}
Again, the fading memory is the continuity in this topology.
\begin{definition}
        A functional $H$ taking inputs in $I^\Delta$ fulfills the fading memory condition if it is continuous in the topology of the metric space $\left(I^K_-, d^c_\omega\right)$.
    \end{definition}
\section{Universality of classical reservoir computers}
The definitions and the theorem above are exploited in the following, as they provide the minimal key ingredients for constructing classes of universal reservoir computers. First, the two classical paradigms of reservoir computing, namely liquid state machine and echo state networks \cite{mass.univ, grig2017} are considered, and next the unified paradigm is applied to an example of a quantum reservoir \cite{naka.2019}. 
\subsection{Universality of classical echo state networks}
Fix $T = \z$ and denote with $I^K = S^{\z,\ns_1}_{K}$ the set of uniformly bounded orbits on the normed space $\ns_1$. An echo state network (ESN) is a system described by the following equations
\begin{equation}\label{esn}
    \begin{cases}
        x_{t+1} = F(x_n, u_{t+1}) \\
        y_{t+1} = h(x_{t+1}),\qquad \forall t\in \z
    \end{cases}
\end{equation}
with $F\colon\ns_2\cross\ns_1\rightarrow\ns_2$ and $h\colon\ns \rightarrow\r$, where $\ns_2$ denotes the normed space of the reservoir states. 
\begin{definition}\label{esp}
     An echo state network has the \textbf{echo state property} if, for any $u\in I^K$, there exists a unique $x \in (\ns_2)^\z$ that solves the equations \eqref{esn}.
\end{definition}
An ESN with echo state property is included in the definition of an echo reservoir computer.
\begin{lemma}\label{esn.is.erc}
    An ESN with echo state property is an echo reservoir computer. Namely, there exists a causal, time-invariant filter $E$ such that Eq. \eqref{esn} can be written as
    \begin{align*}
        \begin{cases}
        x_{t+1} = (E u)_{t+1} \\
        y_{t+1} = h(x_{t+1})
    \end{cases}
    \end{align*}
    In particular, it follows that any ESN with echo state property defines a unique functional $R_{(E,h)}$.
\end{lemma}
\begin{proof}
The existence and uniqueness of $E$ are a straightforward consequence of the echo state property. Moreover, $E$ is causal by construction. Eventually, the uniqueness of reservoir induces  time-invariance, namely $(E z)_{t+\tau} = (E z^\tau)_n$ for any $t \in \z$ and fixed $\tau$. Formally, let $\tilde{x} = E z^\tau$ the reservoir state associated with the time-delayed orbit $z\tau$. By definition, it solves the equation
\begin{equation*}
    \tilde{x}_{t+1} = F(\tilde{x}_{t}, u^\tau_{t+1}) = F(\tilde{x}_{t}, u_{t+1-\tau})
    \end{equation*}
Since by definition $x_{t+1-\tau} = F(x_{t-\tau}, u_{t+1-\tau})$ then, by uniqueness of the solution of Eq. \eqref{esn}, one has $\tilde{x}_{t} = x_{t-\tau}$ for any $t\in \z$. This implies $(E z)_{t+\tau} = (E z^\tau)_t$.
\end{proof}
In Ref. \cite{grig2018} it is identified a ready-to-use condition that assures the echo state property of an ESN and the fading memory of the associated functional. We report this lemma for completeness since we use it later on in our developments.
\begin{lemma}\label{contraction}
   If the reservoir function $F$ is a contraction, namely
\begin{align*}
     \norm{F(x, u) - F(y,u)}\le r \norm{x-y}
\end{align*}
for some $r\in(0,1)$, then the associated ESN has the echo state property. Moreover, there exists a unique functional with fading memory associated with it. 
\end{lemma}
We consider now a specific echo state network architecture and we show that the class of functionals associated with this class is a polynomial algebra, thus we show that this class is universal. We consider echo state networks of the form
\begin{align}\label{esn.ese}
        \begin{cases}
        x_{t+1} = \sigma\left(A x_t + B u_{t+1} + \xi \right)\\
        y_{t+1} = p(x_{t+1})
    \end{cases}
    \end{align}
    with $\sigma$ defined by the component-wise application of any squashing function (i.e the hyperbolic tangent), $A$ a $N\cross N$ matrix for some $N\in\n$ - the dimension of the reservoir -, $B$ a $N\cross n$ matrix and $\xi \in \r^N$. Here $n$ and $N$ are respectively the dimension of the input orbit and the dimension of the reservoir. We denote with $H^{A,B}_p$ the reservoir functional associated with an instance of the systems described by Eq. \eqref{esn.ese}. The echo state property and the fading memory of the associated functional follows depending on the spectral properties of $A$. Namely, it holds the following.
    \begin{lemma}
        Let $\sigma_{max}(A)$ the maximum singular value of $A$ and $L$ the Lipschitz constant of $\sigma$. If $|L \cdot\sigma_{max}(A)|< 1$, then the ESN in \eqref{esn.ese} has the echo state property and the associated functional has fading memory.
    \end{lemma}
    \begin{proof}
    Following Lemma \ref{contraction}, it suffices to prove that $\sigma(A x + B u + \xi )$ is a contraction with respect to $x$. It follows from the hypothesis $|L\cdot \sigma_{max}(A)|< 1$ since 
    \begin{align*}
    \norm{\sigma(A x + B u + \xi ) - \sigma(A y + B u + \xi )}\le \\ \le |L\cdot \sigma_{max}(A)|\norm{x-y}\,.
    \end{align*}
    \end{proof}
    Following Corollary \ref{univer}, the class of reservoir computers of the form \eqref{esn.ese} is universal as long as the associated functionals lay in a polynomial algebra. Given two functionals $H^{A_1,B_1}_p$ and $H^{A_2,B_2}_q$, associated to two given instances of the system in Eq. \eqref{esn.ese}, their sum and product are respectively given by 
$ H^{A_{12}, B_{12}}_{p+q} \text{ and } H^{A_{12}, B_{12}}_{p\cdot q} $
    with\footnotemark$A_{12} = A_1 \oplus A_2$ and $B_{12} = B_1 \oplus B_2$. It remains to verify that the system with reservoir connections $A_1\oplus A_2$ has the echo state property, which follows immediately since $\sigma_{max}( A_1 \oplus A_2 ) = \max\left(\sigma_{max}(A_1), \sigma_{max}(A_2)\right)$. Thus, recalling that $(I_K,d_\omega)$ is a compact metric space and that functionals with fading memory are continuous with respect to $d_\omega$ we have proved the following, applying Theorem \ref{univer}. 
\begin{theorem}
    The class of echo state networks of the form \eqref{esn.ese} with reservoir connections such that $\sigma(A)<1$ is universal, thus it can approximate any fading memory functional defined on $I^K$.
\end{theorem}
\begin{remark}
  Typically, in practical implementations, linear readouts are sufficient to perform many specific tasks. Indeed, even the class of ESN with linear readouts only (so, $p(x_t) = C x_t + D$  in \eqref{esn.ese}) has been proven to be universal \cite{grig2017,grig2018}.  
\end{remark}

\footnotetext{we denote $A\oplus B = \begin{pmatrix}
A & 0 \\
0 & B 
\end{pmatrix}$}
%%%%%%%%%%%%%%%%%%%%%%%%%%%%%%%%%%%%%%%%%%%%%%%%%%%%%%%%%%%%%%%%%%%%%%%%%%%%%%%%%%%
%LIQUID STATE MACHINE 
%%%%%%%%%%%%%%%%%%%%%%%%%%%%%%%%%%%%%%%%%%%%%%%%%%%%%%%%%%%%%%%%%%%%%%%%%%%%%%%%%%%
\subsection{Recasting universality of liquid state machines in the unified framework}
Liquid state machines are echo reservoir computers described by
\begin{equation}\label{lsm}
    \begin{cases}
        x_{t} = (Bu)_t\\
        y_{t} = h(x_{t}),\qquad \forall t \in \r
    \end{cases} 
\end{equation}
with $x_t\in\r^n$ and $B = (B_1,\dots,B_N)$ a N-dimensional causal and time-invariant filter. If the input functions $u$ are continuous, uniformly bounded, and uniformly Lipschitz functions of time, many implementations of the family of filters $B$ and readout $h$ are possible, so that separation of points, constants representation and fading memory are ensured. Then, the corresponding class of liquid state machines is universal, as a consequence of Theorem \ref{univer}. We refer to \cite{mass.univ} for a precise statement. \\
\\
Here we discuss an implementation of an LSM in the spirit of neural microcircuit networks, considering a reservoir composed of a network of $N$ connected integrate-and-fire spiking neurons \cite{maass.nat}. The current liquid state $x_t$ is an N-dimensional vector that represents the contribution of each neuron to the membrane potential at time $t$. Formally, we can describe the associated filters as linear operators with exponential decay (see \cite{maass.nat} for more details) taking spike train inputs $u\in I^\Delta$
\begin{align*}
 B_i u_t \coloneqq \int_{-\infty}^t b_i(s) e^{s-t} \delta_u \, ds   
\end{align*}
where we are denoting with $\delta_u$ the delta function $\delta_u(t) \coloneqq \sum_{s\in u}\delta(t - s)$, and with $b_i(s)$ a continuous, bounded, real function. Such a filter has fading memory with $\omega_t = e^t$. Moreover, it is causal and time-invariant by construction. The readout function $h$ is implemented through a subset of firing neurons that receive connections from any neuron in the reservoir, but that are not mutually connected. Concretely, the current firing activity of this subset of neurons is interpreted as the analogical output of the readout function. It is known that such functions $h$ are universal approximators, namely, they approximate with arbitrary precision any given real function \cite{maass.winner}.
\begin{theorem}
    The class of liquid state machines
    \begin{equation}\label{lsm.teo}
       \begin{cases}
        x_{t} = (Bu)_t = (B_1,\dots,B_N)(u)_t\\
        y_{t} = h(x_{t}),\qquad \forall t \in \r, 
    \end{cases}   
    \end{equation}
    for some $N\in \n$, with $B_i$ linear, exponential decay filter and $h$ any analogical function described above, is universal.
\end{theorem}
\begin{proof}
Fading memory of $B_i$ implies that the functional associated with the LSM in Eq. \eqref{lsm.teo} has fading memory, thus by  Def. \ref{fm-spike}, it is continuous with respect to the metric $d^c_\omega$. The separation property is guaranteed by the definition of the filters since it is enough to choose any function $b$ such that $\sum_{s\in u} e^s b(s) \not = \sum_{s\in v} e^s b(s) $. Then, Lemma \ref{func.dens} ensures that, for any given causal, time-invariant,
 real functional $H\colon I^K\rightarrow\r$ with fading memory, for any $\epsilon>0$ there exists a N-dimensional reservoir filter $B$, for some $N\in\n$, and a polynomial $p$ such that the functional $R_{(B,p)}$ related to the LSM 
  \begin{equation*}
       \begin{cases}
        x_{t} = (Bu)_t = (B_1,\dots,B_N)(u)_t\\
        y_{t} = p(x_{t})
    \end{cases}   
    \end{equation*}
    satisfies $|R_{(B,p)} u - Hu|<\epsilon$ for any $u \in I^\Delta_-$. Eventually, one can replace $p$ with an analog readout function $h$ of the type discussed above, since they are universal approximators of real functions.
\end{proof}
\section{Quantum echo state networks and universality}
The last Section applies the concept of universality to quantum reservoirs based on qubit registers \cite{giappo.1, naka.2019}. Such extension exploits the density matrix representation of a quantum state and involves spatial multiplexing. An effective implementation of provably universal quantum reservoir systems, exploiting dissipative quantum systems, was proposed in Refs. \cite{chen1,chen2}.
\subsection{Quantum echo state networks}
Let $\cH$ be a Hilbert space, that we assume to be finite-dimensional for the sake of simplicity, and let $\cB(\cH)$ be the space of bounded operators on $\cH$, equipped with the Schatten norm $\norm{T} = \text{Tr}( A^\dag A)$. Any pure or mixed state in $\cH$ may be represented by a density operator $\rho\in \cB(\cH)$, namely, by an element in the compact convex subspace
\begin{gather*}
    \cS(\cH) = \left\{\rho \in \cB(\cH): \text{Tr}(\rho)=1, \rho \ge 0,\, \rho = \rho^\dag \,\right\}.
\end{gather*}
At each time step, the reservoir state is represented by a density operator defined on a suitable Hilbert space. The quantum dynamics of the reservoir is described by a quantum channel (see \cite{gqi} for a detailed introduction of this concept). 
\begin{definition}
    A \textbf{quantum channel} is any linear map
    \begin{gather*}
        C: \cB(\cH) \rightarrow \cB(\cH)
    \end{gather*}
    which is completely positive and trace-preserving (CPTC map). 
\end{definition}
\begin{remark}
    It is straightforward to notice that $\cS(\cH)$ is closed with respect to the action of any quantum channel. Namely, the action of any quantum channel $C$ restricts to
\begin{gather*}
        C: \cS(\cH) \rightarrow \cS(\cH)\,.
    \end{gather*}
\end{remark}
A quantum echo state network (qESN) is thus generally described by the following equation
\begin{equation}\label{qrc}
    \begin{cases}
        \rho_{t+1} = C(\rho_t, u_{t+1}) \\
        y_{t+1} = h(\rho_{t+1}),\qquad \forall t\in \z
    \end{cases}
\end{equation}
with, using the terminology introduced in Section \ref{termino}, $(u_t)_{t\in\z}$ and $(y_t)_{t\in\z}$ discrete time orbits in some normed spaces and $(\rho_t)_{t\in\z}$ an orbit in the normed space $\cS(\cH)$ . \\ \\
The notions of echo state property and fading memory essentially replicate those for classical echo state networks (see Definitions \ref{fad.ts} and \ref{esp}) since ESN and qESN ultimately share the same mathematical description. 
In particular, these properties are ensured by the strict contractivity of the CPTP map $C$ that appears in \eqref{qrc}, namely
\begin{align}\label{rle}
 \norm{C(\rho_1, u) - C(\rho_2, u) } \le r \norm{\rho_1 - \rho_2}   
\end{align}
for some $r<1$, as a consequence of Lemma \ref{contraction}.  In this regard, we recall that any CPTP map is non-expansive, meaning that Eq. \ref{rle} holds with $r\le 1$ (see Theorem 9.2 in \cite{nielsen2010quantum}).
With the same argument developed in Lemma \ref{esn.is.erc}, one concludes that a qESN with echo state property falls in our definition of echo reservoir computer. 
\subsection{A class of quantum reservoir computers}\label{qESN.sub}
In these last two subsections, we briefly present an embodiment of a qESN exploiting a $N$-qubit register proposed in \cite{naka.2019} and revised in \cite{chen1,chen2}.  Let $\cH_2$ be the two-dimensional complex Hilbert space associated with a single qubit and $\cH^N = \bigotimes_{j=1}^N \cH_2$ the $2^N$-dimensional Hilbert space of the whole register.  Denoting with $I = \sigma_{00},\, X = \sigma_{01},\, Y = \sigma_{10}, Z = \sigma_{11}$ the Pauli operators, the products $\left\{P_i\right\} = \left\{\bigotimes_{l=1}^N \sigma_{i_{2k-1}i_{2k}}\right\}$ form a basis of the normed space $\cB(\cH^N)$. Then, the density operator $\rho$ that describes the state of the reservoir is represented by a $4^N$-dimensional vector $r = \left(r^{00\dots00},r^{01\dots00},\dots,r^{11\dots11}\right)$ with
\begin{equation*}
    r^i = \frac{1}{2^N}Tr\left[P_i\rho\right]\,.
\end{equation*}
We consider bounded time series $\left\{u_t\right\}_{t\in\z}$ with $ u_t \in  [0,1]$ as inputs of the reservoir system. The inputs are encoded in the dynamics of the reservoir via a quantum channel described at each time step $t$, in the Pauli basis, by a matrix $S_{u_t}$, written as \footnotemark
\begin{align*}
    (S_{u_t})_{ij} = Tr\left[P_j\left(\frac{I + (1-2 u_t)Z}{2}\otimes Tr_1\left(P_i\right)\right)\right]\,.
\end{align*}
\footnotetext{here $Tr_1$ is the partial trace with respect to the first qubit}
The evolution of the reservoir state $r$ is described by the quantum channel
$
 r_{t+1} = S U_H^\tau  r_t
$,
where $U_H^\tau$ is the quantum dynamics induced by some Hamiltonian $H$\footnotemark \footnotetext{formally, one can write$(U^\tau_H)_{ij} = Tr\left[P_je^{-iH\tau}P_i e^{iH\tau} \right]$}. Then, the readout is a linear combination of a subset of some nodes in the reservoir. For example, one can use the $N$ values corresponding to the $Z$ Pauli operators as true nodes. Formally, denoting $z^1 = r^{0100\dots00}, z^2 = r^{0001\dots00},\dots,z^N = r^{0000\dots01}$, the related qESN is described by
\begin{equation}\label{quantum.res}
    \begin{cases}
        r_{t+1} = S_{u_t} U_H^\tau r_t\\
        y_{t+1} = \sum_{i=1}^N w_i z_{t+1}^i 
        \end{cases}
\end{equation}
for some coefficients $w_i \in \r$. We denote each instance of this system as $\cQ_{(H,w)}$. We assume, moreover, that the unitary evolution $U_H^\tau$ is such that the system in Eq. \eqref{quantum.res} is contractive. Thus, the functional associated is well-defined and it has the echo state property and fading memory. 
\begin{table*}[h]
    \centering
    \caption{Comparison between the three paradigms of reservoir computers discussed in the main text and the generalized notion of echo reservoir computer (ERC) defined in this work}
    \begin{tabular}{m{1cm} m{2.6cm} m{2.8cm} m{1.8cm} m{3cm} m{2.8cm}}
     \toprule
        \textbf{Type} & \textbf{Input}  & \textbf{Reservoir} &  \textbf{Time evolution} & \textbf{Readout} & \textbf{Reference} \\
        \midrule
        ESN & Time series in $\r^n$   & Artificial neural network & $T = \z$ & Linear or polynomial function of every node in the reservoir & \cite{jaeger.first.long, grig2017, grig2018},\\
        \addlinespace
        LSM  &  Spike trains & Spiking neural network & $T = \r$ & Analogical function of a subset of neurons & \cite{maass.nat} \\
        \addlinespace
        qESN & Time series in $\r^n$   & Qubit register & $T = \z$ & Linear function of the $Z$ Pauli operators & \cite{naka.2019, chen2}\\ 
        \addlinespace
        ERC & Orbits in any normed space  & Any system in a normed space & $T = \z$ or $T = \r$ & Generic functional applied to the reservoir state & This work\\
    \bottomrule
    \end{tabular}
\end{table*}
\subsection{Universality by spatial multiplexing}
Following Theorem \ref{univer}, we show how the class of spatial multiplexed realizations of qESN generated from the systems described in Subsection \ref{qESN.sub} is universal, as spatial multiplexity allows the representation of sums and products of the associated functionals. Informally, spatial multiplexing consists of preparing uncoupled instances of a reservoir system with different internal parameters and running them in parallel with the same input. More formally, let $\cQ_{(H_1,w)}, \cQ_{(H_2,v)}$ be two independent qESN systems with respectively $N_1$ and $N_2$ qubit registers. Denote with $R_1$ and $R_2$ the respective associated functionals, and with $r^1$ and $r^2$ the respective reservoir states represented in the Pauli basis. Spatial multiplexity provides the construction of the system whose associated functional is the product (or, similarly the sum) of $R_1$ and $R_2$. The reservoir state $r$ of the spatial multiplexed system is the joint reservoir state of each isolated register, namely $r = (r^1, r^2)$. Consequently, we denote the true nodes, corresponding to the $Z$ operators of the two registers, respectively as $z^1,\dots,z^{N_1}$ and $z^{N_1 + 1},\dots,z^{N_1 + N_2}$.
Then, the functional associated with the qESN
\begin{align*}
     \begin{cases}
        r_{t+1} = \left(S_{u_t} U_{H_1}\,r^1_t, \,S_{u_t} U_{H_2}\,r^2_t\right)\\
        y_{t+1} = \left(\sum_{i=1}^{N_1} w_i z_{t+1}^i\right)\cdot\left( \sum_{i={N_1 + 1}}^{N_2} v_i z_{t+1}^i\right)
        \end{cases}
\end{align*}
is the product functional $R_1 \cdot R_2$. Namely, recalling the notation in Eq. \eqref{poly}, for any time series input $\left\{u_t\right\}_{t\in\z}$, one has 
\begin{equation*}
    R_1(u)\cdot R_2(u) =  (R_1 \cdot R_2) (u)\,.
\end{equation*}
The sum is built identically by exploiting $\sum_{i=1}^{N_1} w_i z_{t+1}^i + \sum_{i={N_1 + 1}}^{N_2} v_i z_{t+1}^i$ as readout.
 As long as the unitary dynamics induces fading memory and is able to separate different inputs, we can conclude that the family of quantum reservoir computers that are spatial multiplexed instances of systems described in Eq.\eqref{quantum.res} is universal. In many cases, the proof of separability is a challenging issue. We refer to \cite{chen2} for a detailed discussion on how separability is achieved formally. 
\section{Conclusions}
A unified theoretical framework of reservoir computing is defined and the minimal set of sufficient conditions to ensure that a class of reservoir computers serves as a universal approximator for functionals is demonstrated. Such conditions turn out to be the polynomial algebra structure of the set of associated functionals and their fading memory, respectively. We have shown that such a unified framework not only recovers the two major classical paradigms of reservoir computers, namely the echo state networks and the liquid state machines but also extends to the construction of universal reservoir computers in the arising context of quantum reservoir computing. Guided by our general theorem, we have shown that spatial multiplexing is a computational resource when dealing with quantum reservoirs since it assures universality.
To conclude, echo state computers behave as a single class of universal computing machines including both classical and quantum systems, providing a solid context to develop quantum reservoir computing.

\section*{Declaration of competing interest}
The authors declare that they have no known competing financial interests or personal relationships that could have appeared to influence the work reported in this paper.
\section*{Acknowledgements}
This work was supported by the PRIN-PNRR grant PhysiComp funded by the Italian Ministry of University and Research (MUR). 
\bibliographystyle{elsarticle-num}
\bibliography{refs}
\small{\textbf{Francesco Monzani} received his Ph.D. in Mathematics from the Department of Mathematics at the University of Milan in 2024, with a thesis on infinite-dimensional Hamiltonian dynamical systems. He is currently a postdoctoral researcher in the Department of Physics at the University of Milan, working in the lab led by prof. Enrico Prati. His current research interest lies in the mathematical formalization of recurrent networks, with a focus on quantum and neuromorphic computing.} \\ \\
\small{\textbf{Enrico Prati} received the Ph.D. degree from Università degli Studi di Pisa in 2002.  He has conducted research work in quantum computing and quantum artificial intelligence at National Council of Italy until 2021. He is currently associated professor at Physics Department of Università degli Studi di Milano, teaching quantum computing and artificial intelligence models for physics. He coordinates four research projects on applications of quantum computing.
 He is author or co-author of more than 140 scientific papers published in international journals or conferences, including Nature Quantum Information, Nature Nanotechnology and Nature Communication Physics.}

\end{document}